\documentclass[letterpaper, 10 pt, conference]{ieeeconf}

\usepackage{amssymb,amsmath,color,amsfonts}
\usepackage{graphicx}
\usepackage{caption}
\usepackage{cite}
\IEEEoverridecommandlockouts                              
\newcommand*{\QEDB}{\hfill\ensuremath{\square}}%

\newtheorem{assumption}{Assumption}
\newtheorem{definition}{Definition}

\newtheorem{lemma}{Lemma}
\newtheorem{theorem}{Theorem}

\newtheorem{proposition}{Proposition}

\begin{document}

\title{\bf On Fixed-Time Stability for a Class of Singularly Perturbed Systems using Composite Lyapunov Functions}

\author{Michael Tang\thanks{ M. Tang and J. I. Poveda are with the Department of Electrical and Computer Engineering, University of California, San Diego, La Jolla, CA, 92093. Emails:\{myt001@ucsd.edu;~poveda@ucsd.edu\}.}, Miroslav Krstic\thanks{M. Krstic is with the Department of Mechanical and Aerospace Engineering, University of California, San Diego, La Jolla, CA, 92093. Email:\{krstic@ucsd.edu\}.}, Jorge I. Poveda \thanks{Research supported in part by NSF Grant CMMI No. 2228791 and AFOSR YIP FA9550-22-1-0211. Corresponding author: J. I. Poveda. NSF Grant ECCS Career 2305756}}

\maketitle
\begin{abstract}
Fixed-time stable dynamical systems are capable of achieving exact convergence to an equilibrium point within a fixed time that is independent of the initial conditions of the system. This property makes them highly appealing for designing control, estimation, and optimization algorithms in applications with stringent performance requirements. However, the set of tools available for analyzing the \emph{interconnection} of fixed-time stable systems is rather limited compared to their asymptotic counterparts. In this paper, we address some of these limitations by exploiting the emergence of multiple time scales in nonlinear singularly perturbed dynamical systems, where the fast dynamics and the slow dynamics are fixed-time stable on their own. By extending the so-called composite Lyapunov method from asymptotic stability to the context of fixed-time stability, we provide a novel class of Lyapunov-based sufficient conditions to certify fixed-time stability in a class of singularly perturbed dynamical systems. The results are illustrated, analytically and numerically, using a fixed-time gradient flow system interconnected with a fixed-time plant and an additional high-order example.
\end{abstract}

\section{Introduction}
\subsection{Literature Review}
A variety of complex dynamical systems that emerge in control, learning, and optimization can be decomposed into the interconnection of simpler sub-systems operating on different time scales. A typical approach to assess the stability properties of such systems relies on singular perturbation theory, which studies dynamics of the following form:
\begin{subequations}\label{eq:sys1}
\begin{align}
\dot{x}&=f(x,z,t,\varepsilon)\label{eq:xdynamics},~~~x(t_0)=x_0\\
    \varepsilon \dot{z}&=g(x,z,t,\varepsilon),~~~z(t_0)=z_0,
\end{align}
\end{subequations}
where $\varepsilon$ is a small parameter that induces a time scale separation between the 
``slow'' state $x$ and the ``fast'' state $z$. Singular perturbation tools for the study of dynamical systems of the form \eqref{eq:sys1} were introduced around the 1960's by Vasil'eva, Tikhonov, and Krylov \cite{vasil1978singular,tikhonov1952systems,Krylov1}, and later extended \cite{KokotovicSPBook,khalil}. Recent technical surveys include \cite{SPSurvey2012,HistorySingularPerturbations}.     

Among the various tools available for studying singularly perturbed systems, the composite Lyapunov method, introduced by Saberi and Khalil in \cite{saberi1984quadratic}, has gained widespread popularity due to its versatility. This method leverages Lyapunov functions to examine the stability of both the reduced dynamics and the boundary layer dynamics within the system. In particular, as demonstrated in \cite[Thm. 1]{saberi1984quadratic}, under additional interconnection conditions, if the reduced dynamics and the boundary layer dynamics accommodate a general class of quadratic-type Lyapunov functions, one can establish asymptotic stability for the original ``interconnected" singularly perturbed system, provided that $\varepsilon$ is sufficiently small. Furthermore, when Lyapunov functions adhere to quadratic bounds, it becomes possible to attain exponential stability results as well, as shown in \cite[Thm. 2]{saberi1984quadratic}. These findings, commonly referred to as the \emph{composite Lyapunov method}, have played a crucial role in the analysis and design of controllers and algorithms across various domains, including distributed optimization \cite{kia2013singularly}, extremum seeking control \cite{krstic_book_ESC}, nested control of power systems \cite{subotic2020lyapunov}, and the control of aerospace systems \cite{narang2014nonlinear}.

On the contrary, the range of tools available for investigating ``fixed-time stability" in singularly perturbed systems is rather limited. The concept of fixed-time stability, first introduced in \cite{andrieu2008homogeneous} and extensively explored over the last decade \cite{Fixed_timeTAC,lopez2018finite,lopez2019conditions}, is particularly intriguing due to its capacity to address control \cite{polyakov2023finite}, optimization \cite{fixed_time}, learning \cite{nonsmoothesc,poveda2022fixed}, and estimation \cite{rios2017time} challenges within a predetermined finite time interval that can remain independent of the system's initial conditions. Although there are ample Lyapunov conditions available in the literature for establishing fixed-time stability, the toolkit for examining \emph{interconnected} fixed-time stable systems primarily applies to systems that satisfy specific homogeneity conditions \cite{andrieu2008homogeneous,mendoza2023stability}, or certain conditions resembling linearity, paired with discrete-time dynamics \cite{lei2022event}. This limitation raises the question of whether the composite Lyapunov method can also be employed to investigate fixed-time stability in singularly perturbed systems.

\subsection{Contributions}
In this paper, we provide an affirmative response to the question posed above by introducing a fixed-time stability result for a class of singularly perturbed nonlinear dynamical systems based on the composite Lyapunov method. Specifically, we demonstrate that if the reduced dynamics and the boundary layer dynamics of the original system each possess individual fixed-time Lyapunov functions, then the original system will also exhibit fixed-time stability, provided a set of appropriate interconnection conditions are met and $\varepsilon$ is sufficiently small. Consequently, our result can be regarded as a fixed-time counterpart to the results presented by Saberi and Khalil in \cite{saberi1984quadratic} for asymptotic stability in locally Lipschitz systems. To simplify our presentation and due to space limiations, we focus our attention on time-invariant systems with a right-hand side independent of $\varepsilon$. This model effectively recovers the setting considered in \cite[Ch. 11.5]{khalil}, but we do not assume Lipschitz continuity in the vector fields. To exemplify our results, we present two distinct examples inspired by existing findings and applications of fixed-time stability. First, we investigate fixed-time gradient flows interconnected with a fixed-time stable plant. Such interconnections are prevalent in the analysis and design of various control architectures, where the plant operates at a faster time scale compared to the controller. We demonstrate that when $\varepsilon$ is sufficiently small such interconnections exhibit fixed-time stability when the cost functions are quadratic and strongly convex. This result can be interpreted as a robustness property of the fixed-time gradient flows studied in \cite{fixed_time}  with respect to ``parasitic'' fixed-time stable dynamics. Subsequently, and taking inspiration from \cite{8550571}, we explore a class of second-order fixed-time stable systems interconnected with another fixed-time stable dynamical system that evolves at a faster time scale. We establish that such systems also fulfill our core assumptions, further illustrated through numerical simulations.

The rest of this paper is organized as follows: Section \ref{sec:preliminaries} presents some preliminaries and auxiliary Lemmas. Section \ref{sec:mainresults} presents the main results of the paper. Section \ref{sec:examples} demonstrates how our results can be applied to a fixed-time gradient flow system interconnected with a fixed-time plant. Section \ref{additionalexample} presents an additional high-order example. Finally, Section \ref{sec:conclusions} ends with the conclusions.
\section{Preliminaries}
\label{sec:preliminaries}
\subsection{Notation and Auxiliary Lemmas}
We use $\mathbb{R}_{>0}$ to denote the set of positive real numbers, and $\mathbb{N}$ to denote the set of positive integers. Given a matrix $Q\in\mathbb{R}^{N\times N}$ with real eigenvalues, we define $\overline{\lambda}(Q)$ and $ \underline{\lambda}(Q)$ to be its largest and smallest eigenvalue, respectively. Similarly, we use $\overline{\sigma}(Q)$ and $\underline{\sigma}(Q)$ to denote its largest and smallest singular value, respectively. A continuous function $\rho:\mathbb{R}_{\geq0}\to\mathbb{R}_{\geq0}$ is said to be of class $\mathcal{K}$ if it satisfies $\rho(0)=0$ and is strictly increasing. It is said to be of class $\mathcal{K}_{\infty}$ if it is of class $\mathcal{K}$ and, additionally, it grows unbounded. A function $\beta:\mathbb{R}_{\geq0}\times\mathbb{R}_{\geq0}\to\mathbb{R}_{\geq0}$ is said to be of class $\mathcal{K}\mathcal{L}$ if it is of class $\mathcal{K}_{\infty}$ in its first argument and for each $r>0$, $\beta(r,\cdot)$ is non-increasing and $\lim_{s\to\infty}\beta(r,s)=0$.

\vspace{0.1cm}
\subsection{Fixed-Time Stable Systems}
In this paper, we consider time-invariant dynamical systems of the form
\begin{equation}\label{basic}
    \dot{\zeta}=f(\zeta), \ \ \ \zeta(0)=\zeta_0,
\end{equation}
where $f:\mathbb{R}^N\to\mathbb{R}^N$ is a continuous function, $\zeta\in\mathbb{R}^N$ is the state of the system, and $\zeta_0\in\mathbb{R}^N$ is the initial condition. Note that we do not require $f$ to be differentiable or even Lipschitz. 

\vspace{0.1cm}
\begin{definition}
System \eqref{basic} is said to render the origin $\zeta=0$ \emph{globally finite-time stable} if there exists a (generalized) class $\mathcal{K}\mathcal{L}$ function $\beta$ and a continuous function $T:\mathbb{R}^N\to\mathbb{R}_{\geq0}$ (called the settling time function) such that every solution of \eqref{basic} satisfies:
\begin{equation}
|\zeta(t)|\leq \beta(|\zeta(0)|,t),~~~\forall~t \geq0,
\end{equation}
and $\beta(|\zeta(0)|,t)=0$ for all $t\geq T(\zeta(0))$. System \eqref{basic} is said to render the origin $\zeta=0$ \emph{globally fixed-time stable} if, additionally, there exists $T^*>0$ such that $T(\zeta(0))\leq T^*$ for all $\zeta(0)\in\mathbb{R}^N$. \QEDB 
\end{definition}

\vspace{0.1cm}
The following Lemma corresponds to \cite[Lemma 1]{Fixed_timeTAC}. The converse result is also established in \cite[Thm.2]{8550571} when the settling time $T$ is continuous.

\vspace{0.1cm}
\begin{lemma}\label{definitionfixLyapunov}
Suppose there exists a smooth function $V:\mathbb{R}^N\to\mathbb{R}$ that is positive definite, radially unbounded, and satisfies:
\begin{equation*}
    \dot{V}:=\langle\nabla V(\zeta),f(\zeta)\rangle\le -c_1 V(\zeta)^{p_1}-c_2 V(\zeta)^{p_2},~~\forall~\zeta\in\mathbb{R}^N,
\end{equation*} 
for some $c_1, c_2>0$, $p_1\in(0,1)$ and $p_2>1$. Then, the origin $\zeta=0$ is globally fixed-time stable for the dynamics \eqref{basic}, and the settling time function satisfies
\begin{equation}
T(\zeta_0)\le \frac{1}{c_1(1-p_1)}+\frac{1}{c_2(p_2-1)}
\end{equation}
for all $\zeta_0\in\mathbb{R}^N$. \QEDB 
\end{lemma}
%
%
%
\section{Main results}
\label{sec:mainresults}
In this paper, we consider a sub-class of systems of the form \eqref{eq:sys1}, given by:
\begin{subequations}\label{eq:sys}
\begin{align}
    \dot{x}&=f(x,z)\label{eq:xdynamics}\\
    \varepsilon \dot{z}&=g(x,z),
\end{align}
\end{subequations}
with $\varepsilon>0$, $x\in\mathbb{R}^N, z\in\mathbb{R}^{M}$, and continuous functions $f,g$. We assume that the origin is an equilibrium point for \eqref{eq:sys}, i.e., $f(0,0)=0$ and $g(0,0)=0$. 

\vspace{0.1cm}
Our goal is to study fixed-stability properties of the origin $x=0$ and $z=0$ for system \eqref{eq:sys}, based on the stability properties of a simpler ``reduced'' system that considers $z$ to be at steady state, and the stability properties of the ``boundary-layer'' dynamics that model the initial fast evolution of $z$. In order to introduce these systems, we make the following assumption, which is standard in the literature of singular perturbation theory, see \cite[Ch. 11]{khalil}.

\vspace{0.1cm}
\begin{assumption}\label{assumption:quasimanifold}
There exists a continuously differentiable function $h:\mathbb{R}^N\to\mathbb{R}^M$ such that: 1) $0=g(x,z)$ if and only if $z=h(x)$, for all $x\in\mathbb{R}^N$; 2) $|h(x)|\leq \zeta(|x|)$ for some $\zeta\in\mathcal{K}$, and for all $x\in\mathbb{R}^N$. \QEDB  
\end{assumption}

\vspace{0.1cm}
Using Assumption \ref{assumption:quasimanifold}, we introduce the boundary-layer dynamics of system \eqref{eq:sys}. To do this, we can consider a new state $y = z-h(x)$, which leads to:
\begin{subequations}\label{newvariablesystem}
\begin{align}
    \dot{x}&=f(x,y+h(x)),\\
    \dot{y}&=\frac{1}{\varepsilon}g(x, y+h(x))-\frac{\partial h}{\partial x}f(x, y+h(x)).\label{ydot}
    \end{align}
\end{subequations}
Considering a new time scale $\tau:=t/\varepsilon$, we obtain the following dynamics evolving on the $\tau$-time domain:
\begin{equation}
\frac{\partial y}{\partial\tau}=g(x, y+h(x))-\varepsilon\frac{\partial h}{\partial x}f(x, y+h(x)).
\end{equation}
Setting $\varepsilon=0$, we obtain the \emph{boundary layer dynamics}:
\begin{equation}\label{bl}
    \frac{\partial y}{\partial\tau}=g(x, y+h(x)),
\end{equation}
where $x$ is treated as a fixed parameter. 

Similarly, the \emph{reduced dynamics} of \eqref{eq:sys} are obtained by setting $\dot{z}=0$ and substituting $z=h(x)$ in \eqref{eq:xdynamics}, leading to
\begin{equation}\label{red}
    \dot{x}=f(x,h(x)).
\end{equation}
We now make the following stability assumption on the reduced dynamics and the boundary layer dynamics.
\begin{assumption}\label{assumption1}
There exists a smooth function $V:\mathbb{R}^N\to\mathbb{R}_{\geq0}$ and functions $\alpha_1,\alpha_2\in\mathcal{K}_{\infty}$ such that
\begin{align*}
&~~~~~~~~~~~~\alpha_1(|x|)\leq V(x)\leq \alpha_2(|x|),\\
&\frac{\partial V(x)}{\partial x}f(x,h(x)) \leq -k_1 V(x)^{a_1}-k_2 V(x)^{a_2},
\end{align*}
for all $x\in\mathbb{R}^N$, where $a_1\in(0,1)$, $a_2>1$, $k_1,k_2>0$, and $p$ comes from Assumption \ref{assumption:quasimanifold}.  \QEDB 
\end{assumption}
\begin{assumption}\label{assumptionboundary}
There exists a smooth function $W:\mathbb{R}^N\times\mathbb{R}^M\to\mathbb{R}_{\geq0}$, and functions $\tilde{\alpha}_1,\tilde{\alpha}_2\in\mathcal{K}_{\infty}$ such that:
\begin{align*}
&~~~~~~~~~~~\tilde{\alpha}_1(|y|)\leq W(x,y)\leq\tilde{\alpha}_2(|y|),\\
&\frac{\partial W}{\partial y} g(x, y+h(x))\le -\kappa_1 W(x,y)^{b_1}-\kappa_2 W(x,y)^{b_2},\label{redblfxt}
\end{align*}
for all $x\in\mathbb{R}^N$, $y\in\mathbb{R}^M$, where $b_1\in(0,1)$, $b_2>1$, and $\kappa_1,\kappa_2>0$.
\end{assumption}

%

\vspace{0.1cm}
While  Assumptions \ref{assumption1}-\ref{assumptionboundary} imply that both the reduced and the boundary layer dynamics are fixed-time stable, in general, this condition is not sufficient to guarantee that system \eqref{eq:sys} will also be fixed-time stable. To establish this property, we need to study the following additional interconnection terms:
\begin{align*}
        I_1(x,y):&=\frac{\partial V(x)}{\partial x} \Big(f(x, y+h(x))-f(x, h(x))\Big)\\
        I_2(x,y):&=\left(\frac{\partial W(x,y)}{\partial x}-\frac{\partial W(x,y)}{\partial y} \frac{\partial h(x)}{\partial x}\right)f(x, y+h(x)),
\end{align*}
which, in general, cannot be bounded using standard linear or quadratic terms, as in e.g., \cite[Ch. 11.5]{khalil} because $f$ and $g$ are not Lipschitz. Instead, to bound $I_1,I_2$, we introduce the following terms:
\begin{subequations}\label{definitionsproof}
\begin{align}
        \tilde{V}(x)&:=V(x)^{\frac{a_1}{2}}+ V(x)^{\frac{a_2}{2}}\\
        \tilde{W}(x, y)&:=W(x,y)^{\frac{b_1}{2}}+ W(x,y)^{\frac{b_2}{2}}\\
        \underline{k}&:=\min\{k_1, k_2\},~~\underline{\kappa}:=\min\{\kappa_1, \kappa_2\},
\end{align}
\end{subequations}
where the positive constants $k_1,k_2,\kappa_1,\kappa_2$ come from Assumptions \ref{assumption1}-\ref{assumptionboundary}. Using these terms, we can now state our main stability result for the singularly perturbed system \eqref{eq:sys}, which extends the composite Lyapunov method of \cite[Thm. 11.3]{khalil} from asymptotic stability to fixed-time stability. Due to space limitations, all proofs are omitted.

\vspace{0.1cm}
\begin{theorem}\label{thm1}
Suppose that Assumptions \ref{assumption:quasimanifold}-\ref{assumptionboundary} hold, and that
there exist $\chi_1, \delta_1, \chi_2, \delta_2, c_1, c_2\in\mathbb{R}$ such that: 
\begin{enumerate}[(a)]
\item For all $x\in\mathbb{R}^N$, $y\in\mathbb{R}^M$, $i\in\{1,2\}$: 
\begin{align*}
\hspace{-0.4cm}I_i(x,y)&\le \chi_i \tilde{V}(x)\tilde{W}(x,y)+\delta_i \tilde{V}(x)^2+c_i \tilde{W}(x,y)^2
\end{align*}
\item At least one of the following inequalities holds:
\begin{equation*}
    \delta_1<\frac12 \underline{k},~~~~\text{or}~~~~\delta_2<0.
\end{equation*}
\end{enumerate}

Then, there exists $\varepsilon^*>0$ such that for all $\varepsilon\in(0,\varepsilon^*)$ the origin of \eqref{eq:sys} is fixed-time stable. \QEDB 
\end{theorem}
\section{Fixed-Time Gradient Flows Interconnected With a Fixed-Time Stabilized Plant}
\label{sec:examples}
Here we present an important application of Theorem \ref{thm1} in the context of gradient-flows interconnected with fixed-time stabilized plants. Such types of interconnections (in the context of asymptotic stability) are common across multiple applications in control and real-time optimization.
\subsection{Model and Analysis}
Consider the following fixed-time gradient flow interconnected with a fixed-time plant:
\begin{subequations}\label{gradsp}
    \begin{align}
        \dot{x}&=-k\left(\frac{\nabla \phi(z)}{|\nabla \phi(z)|^{\xi_1}}+\frac{\nabla \phi(z)}{|\nabla \phi(z)|^{\xi_2}}\right)\label{fixedtimeflow}\\
        \varepsilon \dot{z}&=Az+Bu,~~~~u=\sigma(z,x),~ \label{fixedtimefilter}
    \end{align}
\end{subequations}
where $x, z\in\mathbb{R}^N$, $\xi_1\in(0,1)$, $\xi_2<0$, $k>0$, $\nabla\phi$ is the gradient of a cost function $\phi$, and $\sigma$ is a feedback law to be designed. To simplify our presentation, we assume that the matrix $B$ is square and non-singular. The goal is to steer the plant \eqref{fixedtimefilter}, in a fixed time, towards a particular input $u^*=\sigma(x^*)$ that optimizes the cost function $\phi$. This setting describes a standard model-based real-time steady-state optimization problem, drawing inspiration from the ideas in \cite{BIANCHIN2022110579}.

We consider feedback laws of the form $u=u_1+u_2$, which have two main components given by the following smooth and non-smooth terms:
\begin{subequations}\label{fbklaws}
    \begin{align}
        u_1&=-B^{-1}Ax\\
        u_2&=B^{-1}\left(-\nu\frac{z-x}{|z-x|^{\xi_1}}-\nu\frac{z-x}{|z-x|^{\xi_2}}\right),
    \end{align}
\end{subequations}
where $\nu>0$. We assume that the cost $\phi:\mathbb{R}^N\to\mathbb{R}$ is quadratic and has the form:  
$$\phi(x)=\frac12 x^\top Q x+b^\top x+c,$$
where $Q\succ 0$ and $c\in\mathbb{R}$. Without loss of generality, we can disregard the linear term $b^\top x$ since the system can be transformed in a way such that the minimum of $\phi$ lies at the origin. In this way, our goal is equivalent to stabilizing, in a fixed time, the origin of system \eqref{gradsp}.
%
%

%
The following proposition is proved via a sequence of lemmas that show that all the assumptions of Theorem \ref{thm1} are satisfied by system \eqref{gradsp} under the feedback laws \eqref{fbklaws}.

\vspace{0.1cm}
\begin{proposition}
Consider the closed-loop system \eqref{gradsp}-\eqref{fbklaws} with $\nu>\frac{\overline{\sigma}(QA)}{\underline{\lambda}(Q)}$, $k>0, \ \xi_1\in(0,1)$ and $\xi_2<0$. Then, there exists $\varepsilon^*>0$ such that for all $\varepsilon\in(0,\varepsilon^*)$, the origin is fixed-time stable. \QEDB 
\end{proposition}

\vspace{0.1cm}
\textsl{Proof:} Using $y=z-x$, we obtain the following boundary layer dynamics (in the $\tau$-time scale) for system \eqref{gradsp}:
\begin{equation*}
    \frac{dy}{d\tau}=Ay-\nu\left(\frac{y}{|y|^{\xi_1}}+\frac{y}{|y|^{\xi_2}}\right)
\end{equation*}
We can establish the following Lemma, which follows directly by computation.
\begin{lemma}
The reduced system and the boundary layer system satisfy Assumptions \ref{assumption1}-\ref{assumptionboundary} with
\begin{equation}\label{lyapunovfunctions}
V(x)=\frac{1}{2}x^\top Q x,~~~W(y)=\frac12 y^\top Q y,
\end{equation}
and constants 
\begin{align*}
k_1&=2^{1-\frac{\xi_1}{2}}k\underline{\lambda}(Q)^2 \overline{\lambda}(Q)^{-1-\frac{\xi_1}{2}}\\
k_2&=2^{1-\frac{\xi_2}{2}}k\underline{\lambda}(Q)^{2-\xi_2} \overline{\lambda}(Q)^{-1+\frac{\xi_2}{2}}\\
\kappa_i&=2^{1-\frac{\xi_i}{2}}\overline{\lambda}(Q)^{-1+\frac{\xi_i}{2}}(\nu\underline{\lambda}(Q)-\overline{\sigma}(QA))
\end{align*}
and $a_i=b_i=1-\frac{\xi_i}{2}$. \QEDB 
\end{lemma}

\vspace{0.1cm}
Next, to verify that the interconnection conditions of Theorem \ref{thm1} are satisfied by system \eqref{gradsp}, we introduce the following functions $\Upsilon_i:\mathbb{R}^N\times\mathbb{R}^N\to\mathbb{R}$, for $i\in\{1, 2\}$:
\begin{subequations}\label{upi}
\begin{align}
    \Upsilon_i(x,y)&=x^\top\left(\frac{x}{|x|^{\xi_i}}-\frac{y+x}{|y+x|^{\xi_i}}\right)\\
    \Upsilon(x,y)&=\Upsilon_1(x,y)+\Upsilon_2(x,y).
\end{align}
\end{subequations}
The following Lemma will be instrumental for our results:

\vspace{0.1cm}
\begin{lemma}\label{lem1}
Let $\xi_1\in(0,1)$ and $\xi_2<0$. Then, the following inequalities hold for all $x, y\in\mathbb{R}^N$:
    \begin{itemize}
    \item $\left\lvert\Upsilon_1(x,y)\right\rvert\le 2^{\xi_1}|x||y|^{1-\xi_1}$.
    \item
    $\left\lvert\Upsilon_2(x,y)\right\rvert\le \Delta(\xi_2)|x||y|\left(|x|^{-\xi_2}+|y|^{-\xi_2}\right)$,
    \end{itemize}
    where $\Delta(\xi_2)=1+\max\left(1,-\frac{\xi_2}{2^{\xi_2+1}}\right)$.\QEDB
\end{lemma}

\vspace{0.1cm}

\vspace{0.1cm}
Using the functions \eqref{upi}, the interconnection terms of \eqref{gradsp} can be written as:
\begin{subequations}
\begin{align*}
I_1&=k\Upsilon(Qx, Qy)\\
I_2&=-k\Upsilon(Qy, Qx)+k\left(|Qy|^{2-\xi_1}+|Qy|^{2-\xi_2}\right).
\end{align*}
\end{subequations}
With Lemma \ref{lem1}, we can bound $|I_1|$ as follows:
\begin{align}
    |I_1|&\le k\left(|\Upsilon_1(Qx, Qy)|+|\Upsilon_2(Qx, Qy)|\right)\notag\\&\le N_1|x||y|^{1-\xi_1}+N_2|x||y|\left(|x|^{-\xi_2}+|y|^{-\xi_2}\right).\label{i1bd}
\end{align}
Similarly, using Lemma \ref{lem1} we can bound $|I_2|$ as follows: 
\begin{align}
    |I_2|&\le k\Big(|\Upsilon_1(Qy, Qx)|+|\Upsilon_2(Qy, Qx)|\notag\\
    &~~~~~~~~+|Qy|^{2-\xi_1}+|Qy|^{2-\xi_2}\Big)\notag\\&\le N_1|y||x|^{1-\xi_1}+N_2|x||y|\left(|x|^{-\xi_2}+|y|^{-\xi_2}\right)\notag\\&~~~+k\left(\overline{\lambda}(Q)^{2-\xi_1}|y|^{2-\xi_1}+\overline{\lambda}(Q)^{2-\xi_2}|y|^{2-\xi_2}\right),\label{i2bd}
\end{align}
where $N_1=2^{\xi_1}\overline{\lambda}(Q)^{2-\xi_1}k, \ N_2=\Delta(\xi_2)\overline{\lambda}(Q)^{2-\xi_2}k$. To continue, we will state another useful result.

\vspace{0.1cm}
\begin{lemma}\label{lem2}
    Given $p_1, p_2>0$ we can find $\underline{\alpha}, \overline{\alpha}\in\mathcal{K}_\infty$ such that the following holds $\forall x, y\in\mathbb{R}, q>0$:
    \begin{equation}\label{amgm}
        |x|^{p_1} |y|^{p_2}+|x|^{p_2} |y|^{p_1}\le \frac{1}{\underline{\alpha}(q)}|x|^{p}+\overline{\alpha}(q)|xy|^\frac{p}{2}+\frac{1}{\underline{\alpha}(q)}|y|^{p}
    \end{equation}
    where $p:=p_1+p_2.$\QEDB
\end{lemma}

\vspace{0.1cm}
The next Lemma follows directly by computation:

\vspace{0.1cm}
\begin{lemma}
Let $\tilde{V}$ and $\tilde{W}$ be given by \eqref{definitionsproof}, where $V$ and $W$ come from \eqref{lyapunovfunctions}. Then, the following inequalities hold:
\begin{subequations}
\begin{align*}
    &\tilde{V}(x)^2\ge r_1|x|^{2-\xi_1}+r_2|x|^{2-\xi_2}+r_3|x|^{2-\frac12(\xi_1+\xi_2)}\\
     &\tilde{W}(x, y)^2\ge r_1|y|^{2-\xi_1}+r_2|y|^{2-\xi_2}+r_3|y|^{2-\frac12(\xi_1+\xi_2)}\\
     &\tilde{V}(x)\tilde{W}(x, y)\ge r_1|x|^{1-\frac{\xi_1}{2}}|y|^{1-\frac{\xi_1}{2}}+r_2|x|^{1-\frac{\xi_2}{2}}|y|^{1-\frac{\xi_2}{2}}\notag\\&~~+\frac{r_3}{2}\left(|x|^{1-\frac{\xi_1}{2}}|y|^{1-\frac{\xi_2}{2}}+|x|^{1-\frac{\xi_2}{2}}|y|^{1-\frac{\xi_1}{2}}\right)
 \end{align*}
\end{subequations}
where $r_1=2^{\frac{\xi_1}{2}-1}\underline{\lambda}(Q)^{1-\frac{\xi_1}{2}}, r_2=2^{\frac{\xi_2}{2}-1}\underline{\lambda}(Q)^{1-\frac{\xi_2}{2}}$ and $ r_3=2^{\frac14(\xi_1+\xi_2)}\underline{\lambda}(Q)^{1-\frac14(\xi_1+\xi_2)}$. \QEDB 
\end{lemma}

Now we are ready to show that system \eqref{gradsp} satisfies the conditions of Theorem \ref{thm1}. Fix $\xi_1$ and $\xi_2$, and let $\underline{\alpha}, \overline{\alpha}$ be obtained such that items 1) and 2) of Lemma \ref{lem2} hold for $p_1=1,\ p_2=1-\xi_i$ for $i=1,2$. It is possible to find a single pair $\underline{\alpha}, \overline{\alpha}$ that satisfies both cases ($i=1, 2$) since we could apply Lemma \ref{lem2} to both cases separately and take the $\min, \ \max$ of the obtained $\underline{\alpha},\ \overline{\alpha}$ functions respectively.
Let $\mu\in(0, \frac12 \underline{k})$, and choose $q>0$ such that 
$$\frac{1}{\underline{\alpha}(q)}< \min\left(\frac{\mu r_1}{N_1}, \frac{\mu r_2}{N_2}\right):=\eta.$$
Using Lemma \ref{lem2}, we can upper-bound \eqref{i1bd} as follows:
\begin{align*}
    |I_1|&\le \frac{1}{\underline{\alpha}(q)}\Big[(N_1|x|^{2-\xi_1}+N_2|x|^{2-\xi_2})+(N_1|y|^{2-\xi_1}\\&~~~+N_2|y|^{2-\xi_2})\Big]+N_1 \overline{\alpha}(q) |x|^{1-\frac{\xi_1}{2}}|y|^{1-\frac{\xi_1}{2}}\\&~~~+N_2 \overline{\alpha}(q) |x|^{1-\frac{\xi_2}{2}}|y|^{1-\frac{\xi_2}{2}} \\
    &<\mu \tilde{V}(x)^2+\mu \tilde{W}(x,y)^2+\frac{\overline{\alpha}(q)\mu}{\eta}\tilde{V}(x)\tilde{W}(x,y).
\end{align*}
\begin{figure}[t!]
  \centering
    \includegraphics[width=0.42\textwidth]{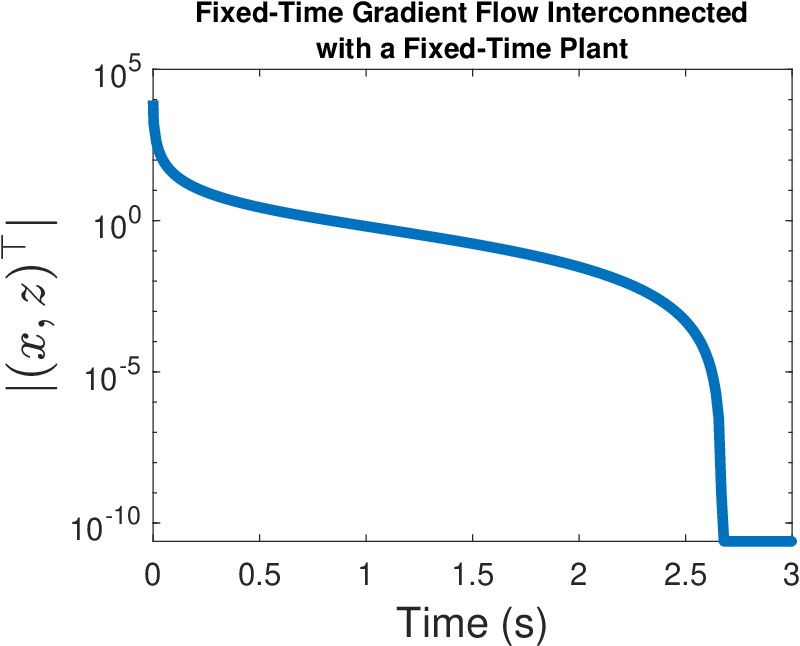}
    \caption{Trajectories of system \eqref{gradsp} with $\varepsilon=0.001$.} \label{fxt1plot}
    \vspace{-0.4cm}
\end{figure}
Similarly, with the same choice of $q$ we can continue from \eqref{i2bd} by exploiting symmetry and obtain the following:
\begin{align*}
    |I_2|&<\mu \tilde{V}(x)^2+\mu^*\tilde{W}(x,y)^2+\frac{\overline{\alpha}(q)\mu}{\eta}\tilde{V}(x)\tilde{W}(x,y),
\end{align*}
where $\mu^*=\mu+k\max\left(\frac{\overline{\lambda}(Q)^{2-\xi_1}}{r_1},\frac{\overline{\lambda}(Q)^{2-\xi_2}}{r_2}\right)$. Therefore, to use Theorem \ref{thm1}, we set $\delta_1=c_1=\delta_2=\mu$, $c_2=\mu^*$, and  $\chi_1=\chi_2=\frac{\overline{\alpha}(q)\mu}{\eta}$. Since $\delta_1<\frac12 \underline{k}$, the conditions of Theorem \ref{thm1} are satisfied, and thus we conclude that the origin of \eqref{gradsp} is fixed-time stable for positive definite quadratic $\phi$ and $\varepsilon$ sufficiently small.
\subsection{Numerical results}
To illustrate the fixed-time stability properties of system \eqref{gradsp}, we consider a numerical example where the quadratic cost $\phi$ has the form $\phi(x)=\frac12 x^\top Q x$, with $Q=[3,2;3,5]$. The parameters of the dynamics are $\xi_1=\frac{1}{3}, \xi_2=-\frac{2}{3}$, and $k=1$. In this case, Lemma \ref{lem2} holds with $\underline{\alpha}(q)=2q,\ \overline{\alpha}(q)=q$ and we can choose $\mu\in(0, \frac12 \underline{k})$ with $\underline{k}=\min(0.359, 0.453)=0.359$. In particular, we choose $\mu=0.1$, which results in $\eta\approx0.0002$. To obtain $\frac{1}{2q}<\eta$ we let $q=3000$. With these values we obtain $\delta_1=c_1=\delta_2=0.1$, $c_2=\mu^*\approx262.6$, and  $\chi_1=\chi_2=\frac{q\mu}{\eta}\approx 1500000$. We pick $\theta=\frac23$, and the matrix $P$ becomes
$$P\approx\begin{bmatrix}
    .02&-750000\\-750000&\frac{0.09}{2\varepsilon}-87
\end{bmatrix}.$$ 
It can be verified that $P\succ0$ for $\varepsilon\in(0, 10^{-15})$. Note that this is a very conservative estimate of $\varepsilon^*$. Indeed, Figure \ref{fxt1plot} shows that the origin is fixed-time stable for significantly larger values of $\varepsilon$ (in the plot we use $\varepsilon=0.001$). The fact that the Lyapunov-based analysis provides conservative bounds on $\varepsilon$ also emerges in the context of asymptotic and exponential stability \cite[Ch. 11.5]{khalil}. However, the power of Theorem \ref{thm1} is to simplify the stability analysis of the nonlinear system \eqref{gradsp}, and to guarantee the \emph{existence} of a feasible $\varepsilon^*$.
\section{High-Order Nonlinear Example}\label{additionalexample}
To further illustrate the strength Theorem \ref{thm1}, and inspired by \cite[Ex. 2]{8550571}, we consider a second-order nonlinear system with ``fixed-time parasitic'' dynamics, given by
\begin{subequations}\label{example2}
    \begin{align}
        \dot{x}_1&=-\bigl \lceil x_1\bigr \rfloor ^{\xi_1}-x_1^3+z\\
        \dot{x}_2&=-\bigl \lceil z\bigr \rfloor ^{\xi_1}-z^3-x_1\\
        \varepsilon \dot{z}&=-\bigl \lceil z-x_2\bigr \rfloor ^{\xi_2}-(z-x_2)^3,
    \end{align}
\end{subequations}
where $\bigl \lceil \cdot\bigr \rfloor^\nu=|\cdot|^\nu \text{sign}(\cdot)$, $x_1,x_2,z\in\mathbb{R}$, $\xi_1, \xi_2\in(0,1)$ and $\xi_1\ge\xi_2$. The parasitic dynamics of this system has a quasi-steady state $h(x)=x_2$.  Using $y=z-h(x)$, we obtain (in the $\tau$-time scale) the boundary layer dynamics:
\begin{equation*}
    \frac{dy}{d\tau}=-\bigl\lceil y\bigr \rfloor^{\xi_2}-y^3,
\end{equation*}
and the following reduced dynamics:
\begin{subequations}
        \begin{equation*}
        \dot{x}_1=-\bigl \lceil x_1\bigr \rfloor ^{\xi_1}-x_1^3+x_2
    \end{equation*}
    \begin{equation*}
        \dot{x}_2=-\bigl \lceil x_2\bigr \rfloor ^{\xi_1}-x_2^3-x_1
    \end{equation*}
\end{subequations}
It can be verified that the Lyapunov functions $V(x)=\frac12 (x_1^2+x_2^2)$ and $W(y)=\frac12 y^2$ satisfy Assumption \ref{assumption1} and inequalities \eqref{redblfxt} with $k_1=k_2=1, a_1=\frac12(\xi_1+1), b_1=\frac12 (\xi_2+1), a_2=b_2=2, \kappa_1=2^{\frac12(\xi_2+1)}$, and $ \kappa_2=4 $. Thus, using \eqref{definitionsproof}:
\begin{subequations}
\begin{align*}
        \tilde{V}(x)&=V^{\frac14(\xi_1+1)}+V\\
        \tilde{W}(y)&=W^{\frac14(\xi_2+1)}+W.
\end{align*}
\end{subequations}
For this system, the interconnection terms are given by
\begin{subequations}
\begin{align}
    I_1&=yx_1-x_2\bigl \lceil y+x_2\bigr \rfloor^{\xi_1}-x_2(y+x_2)^3+x_2^4+x_2 \bigl \lceil x_2\bigr \rfloor^{\xi_1}\label{ex2i1}\\
    I_2&=y\bigl \lceil y+x_2\bigr \rfloor^{\xi_1}+y(y+x_2)^3+yx_1\label{ex2i2}
\end{align}
\end{subequations}
Before we continue, we first need the following Lemma:

\vspace{0.1cm}
\begin{lemma}\label{lemmapowers5}
    For any $x,y\in\mathbb{R}$ and $\xi\in(0,1)$, we have that $x\left(\bigl \lceil x\bigr \rfloor^{\xi}-\bigl \lceil y+x\bigr \rfloor^{\xi}\right)\le 2|x||y|^{\xi}$. \QEDB 
\end{lemma}

\vspace{0.1cm}

\vspace{0.1cm}
By leveraging Lemma \ref{lemmapowers5}, we can expand the term $(y+x_2)^3$ in \eqref{ex2i1} to bound $I_1$:
\begin{align}\label{ex2i1bd2}
    I_1\le |y||x_1|+2|x_2||y|^{\xi_1}+|y|^3 |x_2|+3|y||x_2|^3.
\end{align}
It is then straightforward to verify $|y||x_1|\le 4\tilde{W}(y) \tilde{V}(x)$. We can apply Lemma \ref{lem2} on $|x_2||y|^{\xi_1}$ and $|y|^3 |x_2|+|y||x_2|^3$ to obtain $\underline{\alpha}, \overline{\alpha}\in\mathcal{K}_\infty$ such that:
\begin{equation*}
    |x_2||y|^{\xi_1}\le \frac{1}{\underline{\alpha}(q)}|x_2|^{1+\xi_1}+\overline{\alpha}(q)|x_2 y|^{\frac{1+\xi_1}{2}}+\frac{1}{\underline{\alpha}(q)}|y|^{1+\xi_1}
\end{equation*}
and
\begin{align*}
    |y|^3 |x_2|+|y||x_2|^3\le \frac{1}{\underline{\alpha}(q)}|x_2|^4+\overline{\alpha}(q)|x_2 y|^2+\frac{1}{\underline{\alpha}(q)}|y|^4
\end{align*}
hold $\forall \ q>0$. Let $\mu\in(0,\frac12 \underline{k})$, where $\underline{k}=\min(k_1, k_2)=1$. Choose $q$ sufficiently large such that 
$$\frac{1}{\underline{\alpha}(q)}<\min\left(\frac12\left(2^{-\frac12(\xi_1+1)}\mu\right), \frac13\left(\frac{\mu}{4}\right)\right)=\frac{\mu}{12}.$$
We can then upper-bound \eqref{ex2i1bd2} as follows:
\begin{align*}
    I_1&\le |y||x_1|+2|x_2||y|^{\xi_1}+3\left(|y|^3 |x_2|+|y||x_2|^3\right)\\
    &\le 4\tilde{W}(y) \tilde{V}(x)+\frac{1}{\underline{\alpha}(q)}\left(2|x_2|^{1+\xi_1}+3|x_2|^4\right)\\&~~~+\overline{\alpha}(q)\left(2|x_2 y|^{\frac12(1+\xi_1)}+3|x_2 y|^2\right)\\&~~~+\frac{1}{\underline{\alpha}(q)}\left(2|y|^{1+\xi_1}+3|y|^4\right)\\
    &\le \mu \tilde{V}(x)^2+\mu \tilde{W}(y)^2\\&~~~+\Bigl(\left(2^{\frac14(9+\xi_1)}+12\right)\overline{\alpha}(q)+4\Bigl)\tilde{V}(x)\tilde{W}(y).
\end{align*}
\begin{figure}[t!]
  \centering
    \includegraphics[width=0.42\textwidth]{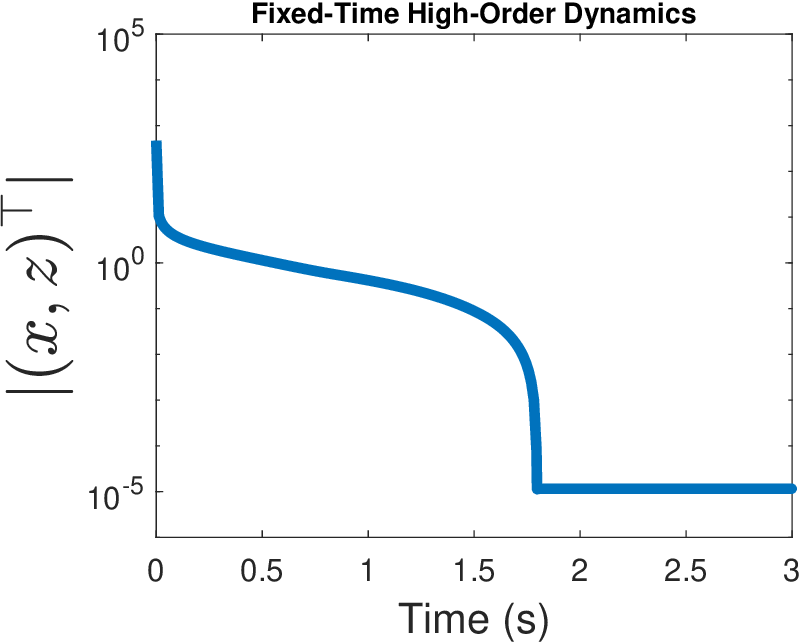}
    \caption{Trajectories of system \eqref{example2} with $\varepsilon=0.001$.}\label{ex2plot}  
    \vspace{-0.4cm}
\end{figure}
Therefore, we conclude that $I_1$ satisfies item (a) of Theorem \ref{thm1}, with $\chi_1=\left(2^{\frac14(9+\xi_1)}+12\right)\overline{\alpha}(q)+4$ and $\delta_1=c_1=\mu$. Note that the condition $\xi_1\ge\xi_2$ is used, as it allows us to state $|y|^\frac{1+\xi_1}{2}\le |y|^\frac{1+\xi_2}{2}+|y|^2\le 2\tilde{W}(y)$.

We will now use a similar technique to bound $I_2$. Using \eqref{ex2i2} we get:
\begin{align*}
    |I_2|&\le |y||y+x_2|^{\xi_1}+y^4+3y^3 x_2+3y^2x_2^2+yx_2^3+yx_1\\&\le |y|^{1+\xi_1}+|y|^4+3|y|^2|x_2|^2+\Big(|y||x_1|+|y||x_2|^{\xi_1}\\&~~~+3|y|^3|x_2|+|y||x_2|^3\Big).
\end{align*}
By exploiting symmetry, we can bound the expression in the parentheses by the same bound used in $I_1$. Also, we have that $|y|^{1+\xi_1}+|y|^4\le 8\tilde{W}(y)^2$ and $3|y|^2|x_2|^2\le 12\tilde{W}(y)\tilde{V}(x)$, which implies that item (a) of Theorem \ref{thm1} is satisfied for $I_2$. Since we have $\delta_1<\frac12 \underline{k}$, the conditions of Theorem \ref{thm1} are satisfied, and thus the origin of \eqref{example2} is fixed-time stable for $\varepsilon$ sufficiently small.

Figure \ref{ex2plot} shows the trajectories of the system \eqref{example2} for the case when $\varepsilon=0.001, \xi_1=\frac13, \xi_2=\frac14, x(0)=(356,241)^\top, z(0)=191$. As observed, the system achieves fixed-time stability.
\section{Conclusions}
\label{sec:conclusions}
We introduce a fixed-time stability result for singularly perturbed dynamical systems based on the composite Lyapunov method. The result establishes that if: 1) the reduced dynamics are fixed-time stable; 2) the boundary-layer dynamics are fixed-time stable 3) the interconnection conditions of Theorem 1 hold, then there exists a sufficiently large time scale separation between the dynamics of $x$ and $z$ such that the origin is fixed-time stable for the interconnected system. Future research directions will explore potential connections between our results and homogeneity properties of the singularly perturbed dynamics. 

\bibliographystyle{IEEEtran}
\bibliography{Biblio,noncoop}
\end{document}